\documentclass[aps,amsmath,amssymb,twocolumn,nofootinbib,pra,superscriptaddress]{revtex4}

\usepackage{graphicx}
\usepackage{dcolumn}
\usepackage{bm}
\usepackage{epsfig,hyperref,amsthm}
\usepackage{mathtools}
\usepackage{color}
\usepackage{colordvi}
\usepackage[dvipsnames]{xcolor}
\usepackage{relsize}
\usepackage{accents}
\usepackage{wasysym}  
\usepackage{xypic}

\newcommand{\ket}[1]{\ensuremath{|#1\rangle}}
\newcommand{\bra}[1]{\ensuremath{\langle #1|}}

\newcommand{\proj}[1]{\ket{#1}\bra{#1}}
\newcommand{\be}{\begin{equation}}
\newcommand{\ee}{\end{equation}}
\newcommand{\ba}{\begin{eqnarray}}
\newcommand{\ea}{\end{eqnarray}}

\newcommand{\tr}[1]{{\rm tr}\left(#1\right)}
\newcommand{\bip}{\mathbb{C}^d\otimes\mathbb{C}^d}
\newcommand{\bipstate}{\mathcal{L}(\mathbb{C}^d\otimes\mathbb{C}^d)}
\newcommand{\norm}[1]{||#1||}
\newcommand{\sbr}[1]{\left[ #1 \right]}

\newcommand{\rIso}{\rho_{\rm iso}}
\newcommand{\id}{\mathbb{I}}

\newcommand{\sys}{{\bf S}}
\newcommand{\env}{{\bf E}}
\newcommand{\bt}{{\bf B}_T}
\newcommand{\obs}{{\bf O}}
\newcommand{\pro}{\mathcal{P}}
\newcommand{\act}{\mathcal{A}}

\newcommand{\F}[1]{\mathcal{F}(#1)}
\newcommand{\Sab}[1]{S(A|B)_{#1}}
\newcommand{\Sba}[1]{S(B|A)_{#1}}
\newcommand{\Wm}[1]{W_{\rm Er}(#1)}
\newcommand{\Wt}[1]{W_{\rm Total}(#1)}

\newtheorem{theorem}{Theorem}[section]

\newtheorem{lemma}[theorem]{Lemma}
\newtheorem{alemma}{Lemma}[section]
\newtheorem{aproposition}[alemma]{Proposition}

\newtheorem{question}{Question}

\newcommand{\update}{\color{black}}
\newcommand{\updatetwo}{\color{black}}
\newcommand{\updatethree}{\color{black}}
\newcommand{\updatefour}{\color{black}}
\newcommand{\updatefive}{\color{black}}
\newcommand{\updatesix}{\color{black}}
\newcommand{\updatesvn}{\color{black}}

\definecolor{nred}{rgb}{0.9,0.1,0.1}
\definecolor{nblack}{rgb}{0,0,0}
\definecolor{nblue}{rgb}{0.2,0.2,0.8}
\definecolor{ngreen}{rgb}{0.2,0.6,0.2}
\definecolor{ublue}{rgb}{0,0,0.5}

\begin{document}
\title{Work extraction and fully entangled fraction}

\author{Chung-Yun Hsieh}
\email{s103022502@m103.nthu.edu.tw}
\affiliation{Department of Physics, National Tsing Hua University, Hsinchu 300, Taiwan}
\author{Ray-Kuang Lee}
\affiliation{Department of Physics, National Tsing Hua University, Hsinchu 300, Taiwan}
\affiliation{Physics Division, National Center for Theoretical Sciences, Hsinchu 300, Taiwan}

\date{\today}

\begin{abstract}
For a bipartite state with equal local dimension $d$, we prove that one can obtain work gain under Landauer's erasure process on one party in identically and independently distributed (iid) limit when the corresponding fully entangled fraction is larger than $\frac{1}{d}$.
By processing a given state to the maximally mixed state, we give an approximation for the largest extractable work with an error in the energy scale, which becomes negligible in the large system limit.
As a step to link quantum thermodynamics and quantum nonlocality, we also provide a simple picture to approximate the optimal work extraction and suggest a potential thermodynamic interpretation of fully entangled fraction for isotropic states.
\end{abstract}

\maketitle

\section{Introduction}
Quantum thermodynamics and quantum nonlocality share similar spirit to tell quantum and classical regimes apart.
For quantum nonlocality, we have the famous Einstein-Podolsky-Rosen (EPR) paradox~\cite{EPR} and Bell's inequality~\cite{Bell} to illustrate the bizarre nature of quantum theory.
On the other hand, for quantum thermodynamics, a multitude of intriguing phenomena related to various definitions of {\em work} has been addressed~\cite{Vinjanampathy2016}.
Examples such as Landauer's principle~\cite{Landauer1961,Alicki2004}, different scenarios on work extraction~\cite{Horodecki2013,Aberg2013,Skrzypczyk2014}, the minimal work cost of running different processes such as a completely-positive trace-preserving map (CPTPM)~\cite{QCI-text, Faist2015}, or an erasure process~\cite{del_Rio2011}, and other related results~\cite{Acin2015, Acin2013, Dahlsten2011,Funo2013,Oppenheim2002,Paraan2009,Fusco2016} simply certify the importance in this direction.

Even though the awareness of the relations between work and entanglement~\cite{Hovhannisyan2013}, coherence~\cite{Korzekwa2016,Misra2016}, or other quantum correlations~\cite{Acin2015,Oppenheim2002,Giorgy2015} have already been studied in the literature, a quantitative connection to nonlocal properties such as quantum nonlocality~\cite{RMP-Bell}, quantum steerability~\cite{Schrodinger,Wiseman2007}, and the usefulness of teleportation~\cite{Bennett1993,Popescu1994,Horodecki1999,Albeverio2002} still remains as an open question deserving further study.
To this end, we adapt a well-known quantity called {\em fully entangled fraction} (FEF)~\cite{Horodecki1999, Albeverio2002}, which acts as a useful measurement in characterizing different nonlocal properties~\cite{RMP-Bell, Zhao2010, Cavalcanti2013, Hsieh2016, Quintino2016}.
In this work, we try to bridge quantum thermodynamics and quantum nonlocality together, by relating work gain under different processes to FEF.

With the help of FEF, for a given bipartite quantum state $\rho$ acting on $\bip$, we derive an inequality to conclude that in identically and independently distributed (iid) limit, it is always possible to have work gain (i.e. negative work cost) of an erasure process on one party of $\rho$, provided that $\rho$ has FEF larger than $\frac{1}{d}$.
This result gives an alternative interpretation for every state useful for teleportation~\cite{Bennett1993, Horodecki1999}, which is equivalent to enabling work gain under local erasure process.
For a specified temperature $T$, we define work extraction to be a process mapping the initial state with fully degenerate system Hamiltonian to the final Gibbs state in $T$ with the Hamiltonian equal to the initial one. 
With some prerequisites, we obtain an approximation for the largest extractable work, up to an error in the energy scale which can be dropped when the system is large enough.
This approximation not only provides a clear strategy for work extraction, which is consisting of two local operations and classical communications (LOCCs) plus one local erasure process, but also suggests us to interpret FEF for isotropic states~\cite{Horodecki1999} as a concept equivalent to the minimal work cost of erasure process on one party.

This paper is structured as follows.
The notations used and the underline scenario to define erasure process and work extraction are described in Sec.~\ref{Sec:Formalism}.
Then, in Sec.~\ref{Sec:Erasure}, the first main result (Theorem~\ref{Result:Work_gain}) is illustrated and proved, which gives an inequality linking FEF and work gain of local erasure process.
The saturation for this inequality is discussed in Sec.~\ref{Sec:Extraction}, as a corollary of the approximation of the optimal extractable work (Theorem~\ref{Result:Estimate}).
Further more, we apply this result to isotropic states and arbitrary pure states in the remaining part of Sec.~\ref{Sec:Extraction}, where the former suggests us a thermodynamic interpretation of FEF, and the latter proposes a nearly-local protocol to extract the maximal work of an arbitrary pure state.
Finally, Sec.~\ref{sec-conclusion} gives the conclusion.
{
\section{PRELIMINARY NOTIONS}\label{Sec:Formalism}
{\updatesvn
Through out this paper, we will only consider quantum states $\rho$ acting on a bipartite system, $\bip$, with equal local dimension $d=2^l$ for some natural number $l\in\mathbb{N}$.
For such a system, we call the first party ``Alice'' and the second party ``Bob''.
The collection of all quantum states acting on this bipartite system is denoted as $\bipstate$.
The {\em conditional von Neumann entropies} (see, for example, Ref.~\cite{RennerPhD}) of a given state $\rho\in\bipstate$ are:
\begin{eqnarray}
&&\Sab{\rho}\coloneqq S(\rho) - S(\rho_{\rm B});\nonumber\\
&&\Sba{\rho}\coloneqq S(\rho) - S(\rho_{\rm A});
\end{eqnarray}
where $S(\rho)\coloneqq - \tr{\rho\log_2{\rho}}$ is the {\em von Neumann entropy} of the quantum state $\rho$, and $\rho_{\rm A}\coloneqq{\rm tr}_{\rm B}(\rho)$ [$\rho_{\rm B}\coloneqq{\rm tr}_{\rm A}(\rho)$] is the local state on Alice's (Bob's) side. 

The {\em fully entangled fraction} (FEF)~\cite{Horodecki1999,Albeverio2002} of a given state $\rho\in\bipstate$ is defined by
\begin{eqnarray}
\F{\rho}\coloneqq\max_{\ket{\Psi}}\bra{\Psi}\rho\ket{\Psi},
\end{eqnarray}
where the maximization is taken over all maximally entangled states $\ket{\Psi}\in \bip$.
It is well known that FEF has been successfully applied to characterize different nonlocal properties.
For instance, a quantum state $\rho\in\bipstate$ is useful for the standard teleportation if and only if $\F{\rho}>\frac{1}{d}$~\cite{Horodecki1999}; while $\rho$ is $k$-copy nonlocal (and hence $k$-copy steerable) if $\F{\rho}>\frac{1}{d}$~\cite{Cavalcanti2013}.
Sufficient conditions for quantum steerability can also be given by FEF~\cite{Hsieh2016, Quintino2016}.
Therefore, if one is able to link FEF with quantum thermodynamics, it can be a good starting point for bridging quantum nonlocality and quantum thermodynamics.

For a given state $\rho\in\bipstate$, we want to consider extractable work from the {\em information content} of $\rho$.
More precisely, we tend to define work extraction by asking {\em how much work can be extracted from $\rho$ by processing it to the maximally mixed state $\frac{\id}{d^2}$?}
The reason  to adapt this definition is because we expect the maximally mixed state to contain the {\em least} useful information content.
Therefore processing $\rho$ to it should be a natural way to obtain the optimal amount of work from the useful information possessed by $\rho$.
{
To study work extraction and other processes, rigorously, we need to specify our framework, which is inspired by Ref.~\cite{del_Rio2011}.
Let us consider the system $\sys=\bip$ couples to an environment $\env$, so that $\sys\otimes\env$ is a closed system with only unitary time evolution.
We require the initial system Hamiltonian to be fully degenerate.
This means any computational basis will be energy eigenbasis; in particular, we are allowed to choose the computational basis in the form $\{\ket{n_1}\otimes\ket{n_2}\otimes ...\otimes\ket{n_l}\,|\,n_1,n_2,...,n_l\in\{0,1\}\}$, where $\{\ket{0},\ket{1}\}$ is a single qubit computational basis.
The environment is consisting of two parts: a heat bath $\bt$ in a finite and positive temperature $T$, and an observer $\obs$.
We further assume the observer possesses a subsystem, called a {\em battery}, which is assumed to be a perfect energy storage to store or withdraw energy without loss. 
The battery is also assumed to be able to generate energy in the form of work.
For example, it can be a collection of work qubits~\cite{Horodecki2013} or a suspended weight that will be raised or lowered~\cite{Skrzypczyk2014,Dahlsten2011}.
In this sense, if energy is extracted from the system by the observer and stored in the battery, one can output the same amount of work from the battery.
Theoretically and ideally, it is in this sense for us to talk about ``work''.

With the above setting, we can now consider the process that may help us to draw work from the system.
For this purpose, we introduce the {\em allowed actions} for the observer $\obs$ to apply, which amounts to observer's abilities to manipulate the system:

(1) {\em Raising/Lowering energy level}. 
This action allows the observer to raise or lower the energy level of the system (e.g. by adjusting the magnetic field), which is a crucial action for work extraction. 
If the system in an energy level $n$ and the observer changes the corresponding energy from $E_n$ to $E'_n$, we define the {\em work gain} (for the observer) to be $E'_n - E_n$~\cite{del_Rio2011}.
Note that we call it work gain based on the assumption of the existence of the ideal battery.
This definition can be generalized to arbitrary states $\rho\in\bipstate$ by ${\rm tr}\sbr{\rho(H'-H)}$, where $H$ is the initial system Hamiltonian, and $H'$ is the final Hamiltonian after this action, i.e., raising/lowering manipulation.
We use $\act_{\rm R/L}$ to denote the set of all actions of this type.
 
 (2) {\em Thermalization.}
 Within this action, the observer will make (part of) the system couple to the heat bath $\bt$ to achieve thermalization, and then the system will decouple from $\bt$ in the end of this action. 
 This amounts to some CPTPMs on the system, which is mainly used to make the system go to the Gibbs state.
 Although a CPTPM does need work cost~\cite{Faist2015}, we assume that this cost is covered by the heat bath rather than the observer.
 Hence, we assume no energy change of the observer in this action.
 We use $\act_T$ to denote the set of all actions of this type with temperature $T$.
 
 (3) {\em Unitary operations}.
 This action includes mappings in the form $\rho\mapsto\int_{U(d^2)}{U\rho U^\dagger P_UdU}$ for some probability density function $P_U$ over the unitary group $U(d^2)$.
 Also, we allow the observer to select specific results and abandon the rest, resulting in non-zero failure probability.
 Note that unlike the actions of thermalization, this action needs work cost from the observer~\cite{Faist2015}.
 In case of negative work cost, we define it to be work gain (upon success) given the existence of the ideal battery.
 We use $\act_{U}$ to denote the set of all actions of this type.

Apart from the actions of unitary operations, we also allow the following source of probabilistic nature, which is a special action for observer:

(4) {\em $\delta$-approximation}. 
Within this action, the observer is allowed to approximate the original input state by another state which is {\em $\delta$-closed} to it (a state $\rho$ is {\em $\delta$-closed} to another state $\sigma$ if their trace distance~\cite{QCI-text} is upper bounded by $\delta$, i.e., $\frac{1}{2}\norm{\rho-\sigma}_{1}\le\delta$).
This {\em conceptual} action has no real effect on the state,  and it is introduced so that  in our framework one can define allowed processes rigorously and completely.
Note that this step will result in nonzero failure probability.
We use $\act_{\delta}$ to denote the set of all actions with $\delta$-precision of this type.

Now, we are in position to define process in our framework: for a given state $\rho\in\bipstate$, we define an {\em allowed process} for $\rho$, denoted by $\pro_\rho$, to be a finite sequence of actions $\pro_\rho =\{A_n\}_{n=1}^N\subset \act_{\rm R/L}\cup\act_{T}\cup\act_{U}\cup\act_\delta$ such that its behavior on the system can be mathematically written as $A_N\circ A_{N-1}\circ...\circ A_1$, where $A_1$ is acting on $\rho$ with the initial system Hamiltonian $H$ (fully degenerate), and $A_N$ will output a final state with final system Hamiltonian equal to the original one, i.e. $H$.
Note that due to this definition, the observer can apply one and only one action at one time.
To illustrate our idea, let us consider the proof of Theorem I.1 in Ref.~\cite{del_Rio2011} as an example. 
In their proof, in order to construct an erasure process on the given state $\rho$, they firstly choose a subsystem which is $\delta$-closed to a maximally entangled pure state $\ket{\Psi}$, and then they apply work extraction ``for $\ket{\Psi}$'' on the given system.
The work extraction they used for $\ket{\Psi}$ is consisting of actions in $\act_{R/L}\cup\act_{T}$ only (see Fig.~4 in Ref.~\cite{del_Rio2011}).
In our framework, such process can therefore be written as $\pro_\rho = \{ A_n\}_{n=1}^N$, where $A_1\in\act_{\delta}$ and $\{A_n\}_{n=2}^N \subset \act_{R/L}\cup\act_{T}$.
As mentioned previously, action of $\delta$-approximation results in probabilistic nature. 
This can be seen by the failure probability for the above process, which is less than $\delta$ as proved by Supplimentary Lemma IV.3 in Ref.~\cite{del_Rio2011}.

Now, we define $W_A(\rho)$ to be the work gain for an action $A$ on $\rho$ upon success.
Following this definition, we can define {\em work gain} of an allowed process $\pro_\rho=\{A_n\}_{n=1}^N$, denoted by $W(\pro_\rho)$, to be $W(\pro_\rho)\coloneqq\sum_{n=1}^N W_{A_{n}}(\rho)$, upon success of all the involved actions.
This also explains why we define an allowed process to be {\em  a chain of actions} rather than a single mapping: if one define a process to be a mapping, it may be ill-defined physically: there may be different chains of actions with {\em different net work gains} result in the same mapping, which is, however, not the physical property that we expect.


\section{Local Erasure Process and Fully Entangled Fraction}\label{Sec:Erasure}

Consider a given quantum state $\rho\in\bipstate$. Within our framework introduced in Sec. ~\ref{Sec:Formalism}, we say $\pro_\rho^{\rm Er|A}$ is a {\em local (Landauer's) erasure process} for $\rho$ on Alice's side if (1) $\pro_\rho^{\rm Er|A} = \pro_{\proj{\phi}}$, where $\proj{\phi}$ is a purification of $\rho$ (and the additional space needed for the purification is assumed to be included in the space of the observer), and (2) $\pro_{\proj{\phi}}$ maps $\proj{\phi}\mapsto\proj{{\bf 0}}\otimes{\rm tr}_{\rm A}\proj{\phi}$, i.e., the state outside Alice's local system will be preserved. Here, we use the notation $\proj{{\bf 0}}\coloneqq(\proj{0})^{\otimes l}$.
This shows why we call it a {\em local} erasure process.
Also, we only consider {\em minimal deterministic average work cost} in iid limit of erasure process defined as follows~\cite{del_Rio2011} (note that we have a minus sign here because we always reserve positive value for work gain): 
\begin{eqnarray}
&-\Wm {\rho} \coloneqq \inf\{&w\, | \, \exists \{\pro^{\rm Er|A}_{\rho^{\otimes k}}\}_{k=1}^\infty\,{\rm s.t.}\, \nonumber\\
&&\lim_{k\to\infty}P[ -W(\pro^{\rm Er|A}_{\rho^{\otimes k}}) \le kw ] =1 \},\quad\quad\quad
\end{eqnarray}
where $P[ -W(\pro) \le x ]$ is the probability for the process $\pro$ to achieve work {\em cost} $-W(\pro)$ no greater than $x$ (recall the probabilistic nature of allowed actions).
To understand the meaning of this definition, note that only the regime $k\gg 1$ will be important.
Hence, $-\Wm {\rho}\times k$ is the minimal work cost that can be achieved almost with certainty for any large enough $k$, and the probability of failure will go to zero when $k\to\infty$. 
This justifies the interpretation of minimal {\em deterministic} average work cost in iid limit.
From Ref.~\cite{del_Rio2011}, we have the following upper bounds:
\begin{eqnarray}\label{Eq:del_Rio}
-\Wm{\rho} \le \Sab{\rho} k_B T\ln{2}.
\end{eqnarray}
The remarkable consequence of the above inequality is the existence of negative work cost of local erasure process for some entangled states~\cite{del_Rio2011}, which also implies that $\Wm{\rho}>0$ serves as an entanglement witness.
However, how strong the entanglement needs to be in order to have work gain under local erasure process is still unknown. 
This is one example of questions that can be answered by our first main result, which is a direct consequence of the following lemma, whose proof can be found in Appendix~\ref{App:Proof_Lemma}:

\begin{lemma}\label{Lemma}
If $\F{\rho}>\frac{1}{d}$, then
\begin{eqnarray}
\Sab{\rho} \le -\log_2⁡{\F{\rho}d}.
\end{eqnarray}
\end{lemma}
}


Combining Eq.~\eqref{Eq:del_Rio} and Lemma~\ref{Lemma}, we obtain our first main result:
}
}
{\updatesvn
\begin{theorem}\label{Result:Work_gain}
If $\F{\rho}>\frac{1}{d}$, then
\begin{eqnarray}
\Wm{\rho} \ge k_B T \ln{\F{\rho}d}.
\end{eqnarray}
In other words, states with FEF larger than $\frac{1}{d}$ enable work gain under local erasure process.
\end{theorem}
}
{\updatethree 
{\updatetwo 

\subsection{Theorem~\ref{Result:Work_gain} and Teleportation}
{\updatesvn
Several remarks can be made accordingly.
As the first one, we recall the famous result developed by Holodecki {\em et al.}~\cite{Horodecki1999} that a given state $\rho\in\bipstate$ is useful for teleportation if and only if $\F{\rho}>\frac{1}{d}$.
Together with Theorem~\ref{Result:Work_gain}, {\updatesvn one learn that }{\em a state $\rho\in\bipstate$ admits deterministic work gain in iid limit under local erasure process if it is useful for teleportation.}
This build a simple while quantitative link between erasure process and the teleportation power.
Also, Ref.~\cite{del_Rio2011} shows that $\Wm{\rho} > 0$ is a {\em sufficient} condition of entanglement. 
Via Theorem~\ref{Result:Work_gain}, one is able to partially address the converse --- we learn that $\F{\rho}> \frac{1}{d}$ is a sufficient condition of $\Wm{\rho}>0$. 
This means when the entanglement is strong enough (in the sense that FEF $>\frac{1}{d}$), there exists a local erasure process inducing work gain.}
}


{\updatetwo
\subsection{Theorem~\ref{Result:Work_gain} and Nonlocal Properties}\label{Sec:Erasure_Nonlocal}
As another remark, Theorem~\ref{Result:Work_gain} can be used to establish sufficient condition of locality and non-steerability {\updatesvn for certain class of quantum states with FEF $>\frac{1}{d}$.
To see this,} we note that Theorem~\ref{Result:Work_gain} implies the following upper bound on the FEF of the given state $\rho$ {\updatesvn with $\F{\rho}>\frac{1}{d}$}:
\begin{eqnarray}\label{Eq:Second_form}
\F{\rho}\le\frac{1}{d}e^{\frac{\Wm{\rho}}{k_B T}}.
\end{eqnarray}
As an example, consider isotropic state~\cite{Horodecki1999} defined by 
\begin{eqnarray}\label{Eq:Iso}
\rho_{\rm iso} (p) \coloneqq p\proj{\Psi_d^+} + (1-p)\frac{\mathbb{I}}{d^2},
\end{eqnarray}
where $\ket{\Psi_d^+}\coloneqq\frac{1}{\sqrt{d}}\sum_{i=0}^{d-1}\ket{ii}$ is the generalized singlet and $p\in[-\frac{1}{d^2 - 1},1]$ due to the positivity of a quantum state. 
{\updatesvn 
As a direct corollary, since $\rho_{\rm iso}$ is entangled if and only if $\F{\rho_{\rm iso}}>\frac{1}{d}$~\cite{RMP-Bell}, Theorem~\ref{Result:Work_gain} implies {\em $\rIso$ admits work gain under local erasure process if and only if $\rIso$ is entangled}.
For other nonlocal correlations, we make use the well-known property of isotropic states: there exist many thresholds for different nonlocal properties~\cite{RMP-Bell}.
For instance, there exist $\mathcal{F}_{\rm LHV}$, $\mathcal{F}_{\rm LHS}^\pi$, and $\mathcal{F}_{\rm LHS}$ such that $\rho_{\rm iso}$ is local under general positive operator-value measures (POVMs) if $\F{\rho_{\rm iso}} \le \mathcal{F}_{\rm LHV}$~\cite{Almeida2007, RMP-Bell}; $\rIso$ is unsteerable under projective POVMs if and only if $\F{\rho_{\rm iso}} \le \mathcal{F}_{\rm LHS}^\pi = \frac{H_d + H_d d - d}{d^2}$, where $H_d\coloneqq\sum_{n=1}^d{\frac{1}{n}}$~\cite{Wiseman2007}, and $\rIso$ is unsteerable under general POVMs if $\F{\rho_{\rm iso}} \le \mathcal{F}_{\rm LHS} = \tilde{p}_\phi (1-\frac{1}{d^2})+\frac{1}{d^2}$, where $\tilde{p}_\phi \coloneqq\frac{3d-1}{d^2-1}(1-\frac{1}{d})^d$~\cite{Almeida2007}. 
Applying Eq.~\eqref{Eq:Second_form} on $\rIso$ with FEF $>\frac{1}{d}$, we learn that it is local if
\begin{eqnarray}\label{Eq:W_LHV}
\Wm{\rho_{\rm iso}} \le  k_B T \ln{\mathcal{F}_{\rm LHV} d};
\end{eqnarray}
it is unsteerable under general POVMs if
\begin{eqnarray}
\Wm{\rho_{\rm iso}} \le  k_B T \ln\sbr{\tilde{p}_\phi \left( d-\frac{1}{d}\right)+\frac{1}{d}};
\end{eqnarray}
it is unsteerable under projective POVMs if
\begin{eqnarray}\label{Eq:W_LHS}
\Wm{\rho_{\rm iso}} \le  k_B T \ln{\frac{H_d+H_d d - d}{d}}.
\end{eqnarray}

}
Because there is a hierarchy consisting of different thresholds for FEF of isotropic states, one may wonder whether Eq.~\eqref{Eq:Second_form} can map this hierarchy onto the one consisting of different erasure work costs.
This can be achieved if the upper bound in Eq.~\eqref{Eq:Second_form} is saturated by isotropic states.
We will back to this issue in Sec.~\ref{Sec:Iso}.

}

}

\section{Work Extraction Process and Fully Entangled Fraction}\label{Sec:Extraction}
{\updatesvn

In this section, we want to study relation between work extraction and fully entangled fraction.
Before proceeding, let us firstly define the former.
Consider a given state $\rho\in\bipstate$.
We say $\pro_\rho^{\rm W}$ is a {\em work extraction process} for $\rho$ if it is an allowed process having $\frac{\id}{d^2}$ as the final state.
Note that the final Hamiltonian is also fully degenerate in our formalism, which means that the maximally mixed state equals  the Gibbs state in the given temperature.
Due to this definition, we define the {\em largest deterministic extractable work} in iid limit for $\rho$ as:
\begin{eqnarray}
&\Wt{\rho}\coloneqq\sup\{&w\,|\,\exists\{\pro_{\rho^{\otimes k}}^{\rm W}\}_{k=1}^\infty\,{\rm s.t.}\nonumber\\
&&\lim_{k\to 1}P[W(\pro_{\rho^{\otimes k}}^{\rm W})\ge kw]=1\}.\quad
\end{eqnarray}
Just like the definition of erasure work cost, $\Wt{\rho}$ serves as the maximal work gain from $\rho$ with certainty in iid limit.
This definition should be an answer of our initial question in iid limit, i.e., how much work can be extracted by processing $\rho$ to $\frac{\id}{d^2}$.

As a direct observation from Eq.~\eqref{Eq:del_Rio}, since the composition of processes $\pro_{(\proj{{\bf 0}})^{\otimes k}}^{\rm W}\circ\pro_{\rho^{\otimes k}}^{\rm Er}$ is a work extraction process 
, we conclude that (see Appendix~\ref{App:W_Total_Proof} for the proof)

\begin{eqnarray}\label{Eq:W_Total_lowerbound}
\Wt{\rho} \ge k_BT\ln{d^2} - S(\rho)k_BT\ln{2}.
\end{eqnarray}

}
{\updatefive
Before stating the main result, we still need to introduce an estimate done by Dahlsten {\em et al.}~\cite{Dahlsten2011}.
Define an {\em $\epsilon$-compression} action for $\rho$ to be an action of unitary operations mapping as $\rho\mapsto U\rho U^\dagger$ such that $U\rho U^\dagger$ is $2\epsilon$-close to a state of the form $\rho' \otimes\proj{\psi}$, with success probability $P_{\rm success}\ge1-2\epsilon$~\cite{Renner2004}.
Then the process adapted by Ref.~\cite{Dahlsten2011} is given by $\{A_n\}_{n=1}^N$, where $A_1$ is an $\epsilon$-compression action, and $\{A_n\}_{n=2}^N$ amounts to a work extraction process on $\proj{\psi}$ with work gain $k_B T\ln{d_\psi}$, where $d_\psi$ is the dimension of the local system $\proj{\psi}$.
We call them {\em compression-extraction} processes, and use the notation $\pro^\epsilon_\rho$ to denote such processes with $P_{\rm success}\ge1-2\epsilon$.
Then} Theorem 2 in Ref.~\cite{Dahlsten2011} states that if the extractable work by a {\updatefour compression-extraction process} of a given bipartite state $\rho\in\bipstate$ is lower bounded by $k_B T \ln{⁡d^2}-\sbr{H_{\rm min}^\epsilon (\rho)+3 \ln{⁡\epsilon}} k_B T \ln{⁡2}$, then we have $P_{\rm success}< 2\epsilon$.
{\updatesix Here $H_{\rm min}^\epsilon (\rho)$ is the smooth min-entropy of $\rho$, whose definition can be found in Eq.~\eqref{Eq:smooth-min-entropy} (we refer the readers to Refs.~\cite{RennerPhD, Dahlsten2011} and references therein for the detail).
This means whenever one is able to choose some $\pro^{\frac{1}{2}-\epsilon}_\rho$ to extract $\Wt{\rho}$ (therefore with $P_{\rm success}\ge 2\epsilon$), Theorem 2 in Ref.~\cite{Dahlsten2011} becomes an {\em upper} bound on it: $\Wt{\rho} < k_B T \ln{⁡d^2}-\sbr{H_{\rm min}^\epsilon (\rho)+ 3 \ln{\epsilon}} k_B T \ln{⁡2}$.
With Lemma~\ref{Lemma} in hand, we are now able to estimate $\Wt{\rho}$ in terms of FEF and conditional von Neumann entropy under the above conditions: [We refer the reader to Appendix~\ref{App:Proof_Estimate} for the proof; {\updatefour also, we adapt the notation $S_{\rm min}(\rho)\coloneqq\min\{S(\rho_{\rm A});S(\rho_{\rm B})\}$}]

\begin{theorem}\label{Result:Estimate}
Given $0 < \epsilon \le \frac{1}{2}$ and $\delta_\epsilon\coloneqq -3\ln{\epsilon}$. 
If \\
{\rm (1)} $\F{\rho}>\frac{1}{d}$,\\
{\rm (2)} there exits $\pro^{\frac{1}{2}-\epsilon}_\rho$ which can extract $\Wt{\rho}$, and \\
{\rm (3)} $\log_2⁡{\norm{\rho}_\infty} = \log_2⁡{\F{\rho} d} - S_{\rm min} (\rho)$, then
\begin{eqnarray}\label{Eq:Picture}
\Wt{\rho} \approx k_B T \ln{d^2} - S_{\rm min}(\rho) k_B T \ln{2} + \Wm{\rho},\,
\end{eqnarray}
and
\begin{eqnarray}
\Wm{\rho}\approx k_B T \ln{\F{\rho}d},
\end{eqnarray}
up to an error $\delta_\epsilon k_B T\ln{2}$.

\end{theorem}


{\updatesvn
As a remark, Eq.~\eqref{Eq:Picture} provides a picture (see Fig.~\ref{Figure:Picture}) of an approximately optimal {\em global} deterministic work extraction in iid limit.}
Note that the above approximation will be faithful if $\delta_\epsilon k_B T\ln{2} \ll k_B T\ln{d}$, i.e., if $\delta_\epsilon \ll l$, where $l$ is the number of qubits in a single party ($d = 2^l$).
In other words, $\delta_\epsilon k_B T\ln{2}$, {\update which is the best resolution for energy in this case}, will be an irrelevant scale. 
Theorem~\ref{Result:Estimate} can be regarded as an approximately sufficient condition of the tightness of Theorem~\ref{Result:Work_gain}, and thus Eq.~\eqref{Eq:del_Rio} and Lemma~\ref{Lemma}.
In particular, approximate saturation for Lemma~\ref{Lemma} together with condition (3) in the above theorem implies the following corollary:
\begin{eqnarray}
S(\rho)\approx-\log_2⁡{\norm{\rho}_\infty},
\end{eqnarray}
up to an error $\delta_\epsilon$.

{\updatesvn

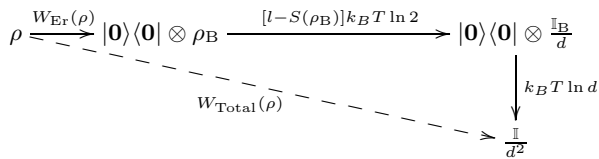
\begin{figure}[t]
\begin{displaymath}
\xymatrix{ \rho \ar[r]^{\Wm{\rho}\quad\quad} \ar@{-->}[drrrr]_{W_{\rm Total}(\rho)\quad} & \proj{{\bf 0}}\otimes \rho_{\rm B} \ar[rrr]^{\, [l - S(\rho_{\rm B})] k_B T \ln{⁡2} } &&& \proj{{\bf 0}}\otimes\frac{\mathbb{I}_{\rm B}}{d} \ar[d]^{k_B T \ln{d}}\\
&&&&\frac{\mathbb{I}}{d^2}
}
\end{displaymath}
\caption{\updatesvn Schematic interpretation of Eq.~\eqref{Eq:Picture}. Assume $S_{\rm min}(\rho)=S(\rho_{\rm B})$ without loss of generality. In this diagram, via an optimal local erasure process in iid limit on Alice's side (the upper-left arrow), a work extraction on Bob's side (the upper-right arrow), and a work extraction on Alice's side (the downward arrow), $\Wt{\rho}$ is extracted (the dashed arrow), up to $\delta_\epsilon k_B T\ln{2}$.}\label{Figure:Picture}
\end{figure}

}

{\update 
To illustrate the applications of Theorem~\ref{Result:Estimate}, let us compute some examples.
In particular, what we want is a thermodynamic interpretation of FEF for certain quantm states. 
}
Due to later consideration, let us define $\Lambda (\rho)\coloneqq\log_2⁡{\norm{\rho}_\infty} - \sbr{\log_2⁡{\F{\rho} d} - S_{\rm min} (\rho)}$.
{\updatesix 
The following subsections hold for $\epsilon$-values satisfying Theorem~\ref{Result:Estimate}.
}

{\updatefive \subsection{Theorem~\ref{Result:Estimate} and Isotropic States}\label{Sec:Iso}}
{\updatesvn 
From Eq.~\eqref{Eq:Iso}, one can see that {\em $\Lambda\sbr{\rIso (p)} = 0$ for all $p$} [to prove this, it suffices to note that $\F{\rIso} = \norm{\rIso}_\infty$ and $S_{\rm min}(\rIso) = \log_2{d}$].
This means for isotropic states with FEF $>\frac{1}{d}$,} if there exists  $\pro^{\frac{1}{2}-\epsilon}_{\rIso}$ which can extract $\Wt{\rIso}$, then $\Wt{\rIso} \approx k_B T \ln{d} + \Wm{\rIso}$ and $\Wm{\rIso} \approx  k_B T \ln⁡{\F{\rIso} d}$ up to an energy scale $\delta_\epsilon k_B T \ln{2}$.
In particular, the latter suggests a possible thermodynamic interpretation of FEF for isotropic states---{\em up to $\delta_\epsilon k_B T \ln{2}$}, FEF of isotropic states is a concept equivalent to the minimal deterministic work cost in iid limit of local erasure process.


{\updatetwo
{\updatesvn Also note that $\Wm{\rIso} \approx  k_B T \ln⁡{\F{\rIso} d}$ builds an approximate hierarchy according to the result discussed in Sec.~\ref{Sec:Erasure_Nonlocal} from the approximate saturation of Eq.~\eqref{Eq:Second_form}; namely, by substituting different FEF thresholds of nonlocal properties for isotropic states into the saturation bound, one obtain the corresponding erasure work cost threshold for isotropic states, up to the energy scale $\delta_\epsilon k_B T \ln{2}$.
For example, if $\rIso$ (with FEF $>\frac{1}{d}$) can achieve the equality in Eq.~\eqref{Eq:Second_form}, then Eq.~\eqref{Eq:W_LHS} implies $\rIso$ is steerable under projective measurement if and only if $\Wm{{\rIso}}\le k_B T \ln{\frac{H_d+H_d d - d}{d}}$.}
Being saturated with a resolution $\delta_\epsilon k_B T \ln{2}$, we may interpret $ k_B T \ln{\frac{H_d+H_d d - d}{d}}$ as an {\em approximate} energy threshold for isotropic states with possible error $\delta_\epsilon k_B T \ln{2}$, which will be small when the system is large enough.
Similar argument applies to other thresholds found in Sec.~\ref{Sec:Erasure_Nonlocal}.

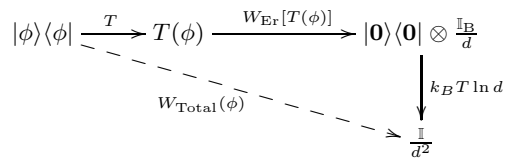
\begin{figure}[t]
\begin{displaymath}
\xymatrix{ \proj{\phi}\ar[r]^{T} \ar@{-->}[drrr]_{\Wt{\phi}\quad} & T(\phi) \ar[rr]^{W_{\rm Er}\sbr{T(\phi)}\quad} && \proj{{\bf 0}}\otimes\frac{\mathbb{I}_{\rm B}}{d}\ar[d]^{k_B T\ln{d}} \\
&&&\frac{\mathbb{I}}{d^2}
}
\end{displaymath}
\caption{\updatesvn Schematic interpretation for the optimal deterministic work extraction in iid limit up to $\delta_\epsilon k_B T \ln{2}$. The upper-left arrow is the quantum twirling bringing the input pure state to an isotropic state, and the remaining processes, i.e. the upper-right and downward arrows, are just the consequence of Eq.~\eqref{Eq:Picture} on isotropic states.}\label{Figure:Pure_state}
\end{figure}

{\updatefive \subsection{Theorem~\ref{Result:Estimate} and Arbitrary Pure States}}
{\updatesvn Now let us consider arbitrary {\updatefour \em pure} states $\ket{\phi}\in\bip$ (we adapt the notation $\phi\coloneqq\proj{\phi}$).
From Ref.~\cite{Horodecki1999}, we learn that} every state $\rho\in\bipstate$ can be turned into an isotropic state via {\em quantum twirling}
\begin{eqnarray}
T(\rho)\coloneqq\int_{U(d)}(U\otimes U^*)\rho(U\otimes U^*)^\dagger dU,
\end{eqnarray} 
where $dU$ is the Haar measure representing uniform distribution over the unitary group $U(d)$. 
{\updatefour Now, the minimal work cost of the quantum twirling on $\phi$ is given by~\cite{Faist2015}} $k_B T \ln{\norm{T(\Pi_\phi)}_\infty}$, where $\Pi_\phi$ is the projector onto the support of the state $\phi$.
Because $\phi$ is pure, we have $\Pi_\phi = \phi$.
Since $\norm{\rIso}_\infty = \F{\rIso}$ and $T(\phi)$ is an isotropic state, the minimal work cost of quantum twirling is 
$k_B T \ln{\mathcal{F}\sbr{T(\phi)}}$.
{\updatesvn By assuming the existence of $\pro^{\frac{1}{2}-\epsilon}_{T(\proj{\phi})}$ which can extract $W_{\rm Total}[T(\proj{\phi})]$,
we use the result for isotropic state to conclude the following approximation for pure states achieving $\mathcal{F}\sbr{T(\phi)}>\frac{1}{d}$:}
\begin{eqnarray}
&\Wt{\phi} &\approx - k_B T \ln{\mathcal{F}\sbr{T(\phi)}} + k_B T\ln{\mathcal{F}\sbr{T(\phi)}d^2}\nonumber\\
&&= k_B T \ln{d^2},
\end{eqnarray}
{\updatesvn up to the precision $\delta_\epsilon k_B T \ln{2}$.
This is the value predicted by Landauer's principle~\cite{del_Rio2011,Oppenheim2002,Alicki2004} (note that, however, this is still an {\em approximation} because the result for isotropic states is not exact).}
{\updatesvn See Fig.~\ref{Figure:Pure_state} for the schematic interpretation.}

\section{Conclusion}
\label{sec-conclusion}
In this work, we try to connect quantum thermodynamics and quantum nonlocality. 
Consider a given state $\rho\in\bipstate$ with dimension $d=2^l$ with $l\in\mathbb{N}$, we prove an inequality which shows that $\rho$ can induce work gain under local erasure process deterministic in iid limit if $\F{\rho}$, its fully entangled fraction (FEF), is larger than $\frac{1}{d}$, thereby connecting work gain under erasure process to the usefulness of quantum teleportation~\cite{Horodecki1999}.

By considering work extraction under temperature $T$ as a process mapping the initial state $\rho$ with a fully degenerate Hamiltonian to Gibbs state  in $T$ with the same Hamiltonian (therefore $\frac{\id}{d^2}$), we derive an approximation for the optimal deterministic extractable work in iid limit, with three prerequisites and an error in the energy scale, which is small in the large system limit. 
The prerequisites of this approximation also serve as a sufficient condition of approximate saturation of our first main result. 
Moreover, a simple picture of optimal work extraction process deterministic in iid limit is proposed by this approximation. 
When it is applicable to isotropic state, we further obtain a possible thermodynamic interpretation of FEF: up to an error in the energy scale, FEF of isotropic states is conceptually equivalent to the minimal work cost (deterministic in iid limit) of local erasure process.
When we focus on pure states, an improved version of the approximation can be derived.
The results we obtained can be a starting point of future research {\updatefour on the interface of quantum thermodynamics and quantum nonlocality}.

\section*{ACKNOWLEDGEMENTS}
This work was supported by the Ministry of Science and Technology of Taiwan under Grant No. 105-2119-M-007-004. Authors acknowledge many fruitful discussions  with Prof. Yeong-Cherng Liang and Prof. Hsiu-Hau Lin.

\appendix

\section{Proof of Lemma~\ref{Lemma}}\label{App:Proof_Lemma}

\begin{proof}
From Ref.~\cite{Konig2009}, we learn that:
\begin{eqnarray}
&&S(A|B)_\rho = \lim_{\epsilon\to 0}\lim_{k\to\infty}\frac{1}{k}H_{\rm min}^\epsilon (A^{\otimes k}|B^{\otimes k})_{\rho^{\otimes k}}\nonumber\\
&&\coloneqq \lim_{\epsilon\to 0}\lim_{k\to\infty}\frac{1}{k}\sup_{\norm{\eta-\rho^{\otimes k}}_B<\epsilon}H_{\rm min} (A^{\otimes k}|B^{\otimes k})_\eta        ,
\end{eqnarray}
where $\norm{\rho-\sigma}_B\coloneqq\sqrt{2-2F(\rho,\sigma)}$ is the {\em Bures distance}~\cite{Konig2009} ($F$ is the fidelity defined in Ref.~\cite{QCI-text}).
According to Theorem 2 in Ref.~\cite{Konig2009}, we have
\begin{eqnarray}
H_{\rm min}(A|B)_\rho = -\log_2[Q(A|B)_\rho d]
\end{eqnarray}
with
\begin{eqnarray}
Q(A|B)_\rho \coloneqq \max_\mathcal{E} \bra{\Psi_d^+}(\mathbb{I}_A\otimes\mathcal{E})(\rho)\ket{\Psi_d^+},
\end{eqnarray}
where $\mathcal{E}:\mathcal{L}(\mathbb{C}^d)\to\mathcal{L}(\mathbb{C}^d)$ is a CPTPM.
{\updatesvn 
Substituting the above form, using Lemma~\ref{Lemma:Conti}, and note the fact $Q(A^{\otimes k}|B^{\otimes k})_{\rho^{\otimes k}}\ge \F{\rho^{\otimes k}}\ge\F{\rho}^k$, we conclude}
\begin{eqnarray}
&&S(A|B)_\rho = -\lim_{\epsilon\to 0}\lim_{k\to\infty}\frac{1}{k}\log_2\sbr{Q(A^{\otimes k}|B^{\otimes k})_{\rho^{\otimes k}} d^k}.\nonumber\\
&&\le -\lim_{\epsilon\to 0}\lim_{k\to\infty}\frac{1}{k}\log_2\sbr{\F{\rho}d}^k = -\log_2{\F{\rho}d}.
\end{eqnarray}
\end{proof}

\section{Notes on the $Q$ function}\label{App:Continuity}
To start with, let us define the {\em Hilbert-Schmidt norm} on an operator $X$ as the norm induced by the Hilbert-Schmidt inner product~\cite{QCI-text}:
\begin{eqnarray}
\norm{X}_2\coloneqq\sqrt{\tr{X^\dagger X}}
\end{eqnarray} 
First, we prove the following lemma:

\begin{alemma}\label{Alemma:Continuity}
$Q(A|B)_\rho$ is continuous on $\rho\in\bipstate$ with Hilbert-Schmidt norm.
\end{alemma}
\begin{proof}
Given $\norm{\rho - \sigma}_2 < \epsilon$, there exists a CPTPM $\mathcal{E}_\epsilon$ achieving [We choose $Q(A|B)_\rho = \bra{\Psi_d^+}(\mathbb{I}_{\rm A}\otimes \mathcal{E}_\epsilon)(\rho)\ket{\Psi_d^+}$ when $Q(A|B)_\rho \ge Q(A|B)_\sigma$ and $Q(A|B)_\sigma = \bra{\Psi_d^+} (\mathbb{I}_{\rm A}\otimes \mathcal{E}_\epsilon)(\sigma)\ket{\Psi_d^+}$ when $Q(A|B)_\rho \le Q(A|B)_\sigma$]
\begin{eqnarray}
&&\left| Q(A|B)_\rho - Q(A|B)_\sigma \right| \nonumber\\
&&\le \left|\bra{\Psi_d^+}(\mathbb{I}_{\rm A}\otimes\mathcal{E}_\epsilon)(\rho)\ket{\Psi_d^+} - \bra{\Psi_d^+}(\mathbb{I}_{\rm A}\otimes\mathcal{E}_\epsilon)(\sigma)\ket{\Psi_d^+}\right| \nonumber\\
&&= \left|\bra{\Psi_d^+}(\mathbb{I}_{\rm A}\otimes\mathcal{E}_\epsilon)(\rho - \sigma)\ket{\Psi_d^+}\right|\nonumber\\
&&= \left| {\rm tr}\sbr{(\mathbb{I}_{\rm A}\otimes\mathcal{E}_\epsilon^\dagger)(\proj{\Psi_d^+})(\rho - \sigma)}\right| \nonumber\\
&&= \frac{1}{d} \left|{\rm tr}\sbr{\mathcal{J}(\mathcal{E}_\epsilon^\dagger)(\rho - \sigma)}\right|,
\end{eqnarray}
where $\mathcal{J}$ is the Choi-Jamio\l kowski isomorphism~\cite{Isomorphism}.
{\updatefour Now, we use the following generalized Cauchy-Schwarz inequality of operators proved by Bhatia, which holds for arbitrary operators $X$, $Y$, and arbitrary unitarily invariant norms $\norm{\cdot}$~\cite{Bhatia1988, Bhatia1995}:
\begin{eqnarray}\label{Eq:General_Cauchy}
\norm{\,|X^\dagger Y|^{\frac{1}{2}}}^2 \le \norm{X}\,\norm{Y},
\end{eqnarray}
{\updatesvn 
where a norm $\norm{\cdot}$ is {\em unitarily invariant} if $\norm{UXV} = \norm{X}$ holds for all operator X and unitary operators $U$, $V$.
Since $\norm{\cdot}_2$ is unitarily invariant,} 
direct computation shows that
\begin{eqnarray}
\norm{X}_2\norm{Y}_2 &&\ge \norm{\,|X^\dagger Y|^{\frac{1}{2}}}^2_2\nonumber\\
&&= {\rm tr}\left[ |X^\dagger Y|^{\frac{1}{2},\,\dagger}|X^\dagger Y|^{\frac{1}{2}}\right] = {\rm tr}|X^\dagger Y|.\quad\quad
\end{eqnarray}
This implies (note that the inequality $|\tr{A}|\le{\rm tr}|A|$ holds for {\em any} operator $A$~\footnote{\updatefour
To prove $|\tr{A}|\le{\rm tr}|A|$ for an arbitrary operator $A$, we firstly choose the polar decomposition~\cite{QCI-text} as $A=UJ$, where $U$ is unitary and $J$ is positive.
Then one can see it remains to prove $|\tr{UJ}|\le\tr{J}$, where $J$ is nothing but $|A|$.
To show this, let us write the spectrum decomposition~\cite{QCI-text} as $J=\sum_{n}{a_n \proj{\phi_n}}$, where $\{\ket{\phi_n}\}$ is an orthonormal basis of the given state space and $a_n \ge 0\ \,\forall\,n$.
Then direct computation shows $|\tr{UJ}| = |\sum_{n,m}\bra{\phi_m}Ub_n\ket{\phi_n}\bra{\phi_n}\phi_m\rangle| = |\sum_n \bra{\phi_n}U\ket{\phi_n}b_n| \le \sum_n |\bra{\phi_n}U\ket{\phi_n}|b_n \le \sum_n b_n = \tr{J}$.
})
\begin{eqnarray}
&&\frac{1}{d} \left|{\rm tr}\sbr{\mathcal{J}(\mathcal{E}_\epsilon^\dagger )(\rho - \sigma)}\right| \le \frac{1}{d}{\rm tr}\left|\mathcal{J}(\mathcal{E}_\epsilon^\dagger )(\rho - \sigma)\right|\nonumber\\
&&\le \frac{1}{d} \norm{\mathcal{J}(\mathcal{E}_\epsilon^\dagger )}_2\norm{\rho - \sigma}_2  \le \frac{1}{d} \norm{\rho - \sigma}_2 < \frac{\epsilon}{d}.
\end{eqnarray}
}Note that $\norm{\mathcal{J}(\mathcal{E}_\epsilon^\dagger )}_2 \le 1$ since $\mathcal{J}(\mathcal{E}_\epsilon^\dagger )$ is a normalized state due to Choi-Jamio\l kowski isomorphism theorem~\cite{Isomorphism}. 
This completes the proof.
\end{proof}

Now, we are in position to prove the following result:

{\updatesvn
\begin{alemma}\label{Lemma:Conti}
If $\F{\rho}>\frac{1}{d}$, then for all $k\in\mathbb{N}$ and $\epsilon\in(0,1)$, there exists $0\le o (\epsilon)\le\log_{2}{(\frac{2\epsilon}{1-2\epsilon} + 1)}$ such that
\begin{eqnarray}
&&\inf_{\norm{\eta - \rho^{\otimes k}}_B < \epsilon} \log_{2⁡}{Q(A^{\otimes k} | B^{\otimes k})_\eta} \nonumber\\
&&= \log_{2⁡}{Q(A^{\otimes k} | B^{\otimes k})_{\rho^{\otimes k}}} + o (\epsilon).
\end{eqnarray}
\end{alemma}
}
\begin{proof}
{\updatesvn 
Consider a fixed $k$ value and a fixed $\epsilon\in(0,1)$.}
We first note the relation $\norm{\rho - \sigma}_2 \le 2\norm{\rho - \sigma}_B$~\footnote{Note that $ \norm{\rho - \sigma}_2\le\norm{\rho - \sigma}_1\le2\sqrt{1 - F(\rho , \sigma)^2}\le2\sqrt{2[1 - F(\rho,\sigma)]} =  2\norm{\rho - \sigma}_B$~\cite{QCI-text}, where $\norm{A}_1\coloneqq{\rm tr}|A|$ is the {\rm trace norm}. To see the fact $\norm{\rho - \sigma}_2\le\norm{\rho - \sigma}_1$, we first apply polar decomposition~\cite{QCI-text} to write $\rho - \sigma = UJ$, where $U$ is a unitary operator and $J$ is a positive operator. Then we have $\norm{\rho - \sigma}_2 = \norm{J}_2$ and $\norm{\rho - \sigma}_1 = \norm{J}_1$. Choose the spectrum decomposition~\cite{QCI-text} as $J = \sum_n a_n\proj{\phi_n}$, where $\{\ket{\phi_n}\}$ is an orthonormal basis and $a_n \ge 0\ \,\forall\,n$. One can verify that $\norm{J}_2^2 = \sum_n a_n^2 \le \left( \sum_n a_n \right)^2 = \norm{J}_1^2$, which proves the desired result.}.
This means $\norm{\eta - \rho^{\otimes k}}_2 < 2\epsilon$ if $\norm{\eta -\rho^{\otimes k}}_B < \epsilon$. 
When $\norm{\eta - \rho^{\otimes k}}_B < \epsilon$, Lemma~\ref{Alemma:Continuity} implies
\begin{eqnarray}
&&\left| Q(A^{\otimes k} | B^{\otimes k})_{\rho^{\otimes k}} - Q(A^{\otimes k} | B^{\otimes k})_\eta \right| \nonumber\\
&&\le \frac{1}{d^k} \norm{\rho^{\otimes k} - \eta}_2 < \frac{2\epsilon}{d^k}.
\end{eqnarray}
{\updatesvn 
For $\rho$ with $\F{\rho}>\frac{1}{d}$, the function $Q(A^{\otimes k} | B^{\otimes k})_{\rho^{\otimes k}}$ only takes values from the interval $(\frac{1}{d^k},1]$.
Hence, if $|x - y| < \frac{2\epsilon}{d^k}$ for some $x\in(\frac{1}{d^k},1]$ and $y\in (\frac{1-2\epsilon}{d^k},1]$ (assume $x\ge y$ without loss of generality), we have $|\log_{2}{x} - \log_{2}{y}| < \log_{2}{\left(\frac{2\epsilon}{d^k y} + 1\right)}< \log_{2}{(\frac{2\epsilon}{1-2\epsilon} + 1)}$.
Substituting $x = Q(A^{\otimes k} | B^{\otimes k})_{\rho^{\otimes k}}$ and $y = Q(A^{\otimes k} | B^{\otimes k})_\eta$ and consider $\norm{\eta -\rho^{\otimes k}}_B < \epsilon$, we conclude
\begin{eqnarray}
&&\inf_{\norm{\eta - \rho^{\otimes k}}_B < \epsilon}\log_{2}{Q(A^{\otimes k} | B^{\otimes k})_\eta} \nonumber\\
&&= \log_{2}{Q(A^{\otimes k} | B^{\otimes k})_{\rho^{\otimes k}}}- o (\epsilon)
\end{eqnarray}
with a function $o (\epsilon)$ satisfying $0\le o (\epsilon)\le\log_{2}{(\frac{2\epsilon}{1-2\epsilon} + 1)}$.
}
\end{proof}



{\updatesvn
\section{Proof of Equation~(\ref{Eq:W_Total_lowerbound})}\label{App:W_Total_Proof}
\begin{proof}
First, we consider a fixed $\epsilon\in(0,1)$.
For a given $k$, we choose $\Delta=\sqrt{k}$ for $\rho^{\otimes k}$ in Supplementary Corollary I.2 in Ref.~\cite{del_Rio2011}, where their $S$ amounts to our system $\sys$ and $Q$ is now a trivial system (with dimension 1) because we want to study work extraction on $\sys$.
Then this corollary implies the existence of a work extraction process $\pro_{\rho^{\otimes k}}^{\rm W}$ such that $\frac{1}{k}W(\pro_{\rho^{\otimes k}}^{\rm W})\ge k_B T \ln{d^2} - [\frac{1}{k}H_{\rm max}^\epsilon(\rho^{\otimes k}) + \frac{1}{\sqrt{k}}]k_B T\ln{2}$, except with a probability of at most $\sqrt{2^\frac{-\sqrt{k}}{2}+12\epsilon}$. 
Here we use the notation $H_{\rm max}^\epsilon(\eta) = H_{\rm max}^\epsilon(S|Q)_\eta$ for the $\epsilon$-smooth max-entropy of $S$ conditional on $Q$~\cite{del_Rio2011}.  
Using the fact $\lim_{k\to\infty}\frac{1}{k}H_{\rm max}^\epsilon(\rho^{\otimes k})= S(\rho)$ $\forall\,\,\epsilon\in(0,1)$, we conclude $\lim_{k\to\infty}P\{\frac{1}{k}W(\pro_{\rho^{\otimes k}}^{\rm W}) \ge k_B T \ln{d^2} - S(\rho)k_B T\ln{2}-\delta \}= 1-\sqrt{12\epsilon}$ for all small enough value $\delta>0$.
In other words, this means for every $\delta>0$, there exists $k_\delta$ such that  $P\{\frac{1}{k}W(\pro_{\rho^{\otimes k}}^{\rm W}) \ge k_B T \ln{d^2} - S(\rho)k_B T\ln{2}-\delta \}> 1-\sqrt{12\epsilon} - \delta$ $\forall\,\, k\ge k_\delta$.
For a fixed $\delta$, we let $\epsilon\to 0$ and obtain $P\{\frac{1}{k}W(\pro_{\rho^{\otimes k}}^{\rm W}) \ge k_B T \ln{d^2} - S(\rho)k_B T\ln{2}-\delta \}> 1 - \delta$ $\forall\,\, k\ge k_\delta$.
This implies $k_B T \ln{d^2} - S(\rho)k_B T\ln{2}-\delta\le\Wt{\rho}$ for all $\delta>0$.
Letting $\delta\to 0$ will result in the desired inequality.
\end{proof}
}

\section{Proof of Theorem~\ref{Result:Estimate}}\label{App:Proof_Estimate}

We first prove the following proposition, which is a general version of Theorem~\ref{Result:Estimate} [recall we define $S_{\rm min}(\rho)\coloneqq\min{\{S(\rho_{\rm A}),S(\rho_{\rm B})\}}$]:

\begin{aproposition}\label{Result:Estimate_general}
Given $0 < \epsilon \le \frac{1}{2}$ and $\delta > \delta_\epsilon\coloneqq -3\ln{\epsilon}$.  
If \\
{\rm (1)} {\updatesvn $\F{\rho}>\frac{1}{d}$,\\
{\rm (2)} there exists $\pro_\rho^{\frac{1}{2}-\epsilon}$} which can extract $\Wt{\rho}$, and \\
{\rm (3)} $- \delta_\epsilon < \log_2{\norm{\rho}_\infty} - \sbr{\log_2{\F{\rho}d} - S_{\rm min}(\rho)} < \delta - \delta_\epsilon$,
\\then we have
\begin{eqnarray}
\left| \frac{\Wt{\rho}}{k_B T\ln{2}} - \sbr{\log_2{\F{\rho}d^3} - S_{\rm min}(\rho)}\right| < \delta.
\end{eqnarray}
\end{aproposition}
\begin{proof}
From Lemma~\ref{Lemma}, Eq.~\eqref{Eq:W_Total_lowerbound}, and Refs.~\cite{Dahlsten2011,Oppenheim2002}, we obtain
\begin{eqnarray}
&&k_B T \ln{⁡d^2} - \sbr{H_{\rm min}^\epsilon (\rho) + 3 \ln{⁡\epsilon}} k_B T \ln{⁡2}\nonumber\\ 
&&> \Wt{\rho} \ge k_B T \ln{⁡d^2} - S(\rho) k_B T \ln{⁡2} \nonumber\\
&&\ge k_B T \ln⁡{\F{\rho} d^3} - S_{\rm min} (\rho) k_B T \ln{⁡2},
\end{eqnarray}
where~\cite{Dahlsten2011}
\begin{eqnarray}\label{Eq:smooth-min-entropy}
H_{\rm min}^\epsilon (\rho) \coloneqq \sup_{\norm{\eta - \rho}_B < \epsilon}⁡ ( - \log_2⁡{\norm{\eta}_\infty}) \ge - \log_2⁡{\norm{\rho}_\infty}.\quad
\end{eqnarray}
This implies
\begin{eqnarray}
&&\log_2{⁡d^2} + \log_2⁡{\norm{\rho}_\infty} - 3 \ln{\epsilon} > \frac{\Wt{\rho}}{k_B T \ln{⁡2}} \nonumber\\
&&\ge \log_2⁡{\F{\rho} d^3} - S_{\rm min} (\rho).
\end{eqnarray}
Hence, a sufficient condition for the inequality $\left| \frac{\Wt{\rho}}{k_B T \ln{⁡2}} - \sbr{\log_2⁡{\F{\rho} d^3} - S_{\rm min} (\rho)}\right| < \delta$ reads:
\begin{eqnarray}
0 < \log_2⁡{\norm{\rho}_\infty} - \sbr{\log_2⁡{\F{\rho} d} - S_{\rm min} (\rho)} + \delta_\epsilon < \delta,
\end{eqnarray}
where $\delta_\epsilon\coloneqq - 3 \ln{\epsilon} > 0$ is the best precision that we may have, since it is the lower bound of all the allowed $\delta$-values.
Note that being larger than $0$ is necessary to guarantee no contradiction.
\end{proof}

{\updatesix We are now in position to prove Theorem~\ref{Result:Estimate}.}
\begin{proof}
To begin with, note that $\Wm{\rho} \ge k_B T \ln⁡{\F{\rho} d}$ {\updatesvn for state $\rho$ with $\F{\rho}>\frac{1}{d}$} from Theorem~\ref{Result:Work_gain}, and $\Wt{\rho} \le \delta k_B T\ln{2} + k_B T \ln{⁡2} \sbr{\log_2⁡{\F{\rho} d^3} - S_{\rm min} (\rho)}$ if it is possible to choose $\pro_\rho^{\frac{1}{2}-\epsilon}$ to extract $\Wt{\rho}$ and $\log_2⁡{\norm{\rho}_\infty} =\log_2⁡{\F{\rho} d} - S_{\rm min} (\rho)$  (hence we can choose $\delta \approx \delta_\epsilon\coloneqq - 3 \ln{\epsilon}$), which is due to Theorem~\ref{Result:Estimate}.
Then direct computation shows [choose $S_{\rm min} (\rho) = S(\rho_{\rm B})$ without loss of generality]
\begin{eqnarray}\label{Eq:Chain}
0 &&\le \Wt{\rho} - \sbr{k_B T \ln{d^2} - S_{\rm min} (\rho) k_B T \ln{⁡2}  + \Wm{\rho}}\nonumber\\
&&\le  \delta k_B T\ln{2} + k_B T \ln⁡{\F{\rho}d} - \Wm{\rho} \nonumber\\
&&\le  \delta k_B T\ln{2} \approx \delta_\epsilon k_B T\ln{2},
\end{eqnarray}
{\updatesvn
where the first inequality follows from the fact that $k_B T \ln{d^2} - S_{\rm min} (\rho) k_B T \ln{⁡2}  + \Wm{\rho}$ can be interpreted as the deterministic work gain (in iid limit) of a particular work extraction process (see Fig.~\ref{Figure:Picture}), thereby being smaller than the optimal one, i.e. $\Wt{\rho}$.
}
This implies, up to the energy scale $\delta_\epsilon k_B T\ln{2}$, the desired results.
\end{proof}

\end{document}